\documentclass[fleqn]{llncs}
\usepackage{version}
\usepackage{tapt}

\title{Thread Algebra for Poly-Threading%
       \thanks{This research has been partly carried out in the
               framework of the Jacquard-project Symbiosis, which is
               funded by the Netherlands Organisation for Scientific
               Research (NWO).}}
\author{J.A. Bergstra \and C.A. Middelburg}
\institute{Programming Research Group, University of Amsterdam, \\
           Kruislaan~403, 1098~SJ~Amsterdam, the Netherlands \\
           \email{J.A.Bergstra@uva.nl,C.A.Middelburg@uva.nl}}

\begin{document}

\maketitle

\begin{abstract}
Threads as considered in basic thread algebra are primarily looked upon
as behaviours exhibited by sequential programs on execution.
It is a fact of life that sequential programs are often fragmented.
Consequently, fragmented program behaviours are frequently found.
In this paper, we consider this phenomenon.
We extend basic thread algebra with the barest mechanism for sequencing
of threads that are taken for fragments.
This mechanism, called poly-threading, supports both autonomous and
non-autonomous thread selection in sequencing.
We relate the resulting theory to the algebraic theory of processes
known as \ACP\ and use it to describe analytic execution architectures
suited for fragmented programs.
We also consider the case where the steps of fragmented program
behaviours are interleaved in the ways of non-distributed and
distributed multi-threading.
\begin{keywords}
poly-threading, thread algebra, process algebra, execution architecture,
non-distributed multi-threading, distributed multi-threading.
\end{keywords}%
\begin{classcode}
D.4.1, F.1.1, F.1.2, F.3.2.
\end{classcode}
\end{abstract}

\section{Introduction}
\label{sect-intro}

In~\cite{BM08b}, we considered fragmentation of sequential programs that
take the form of instruction sequences in the setting of program
algebra~\cite{BL02a}.
The objective of the current paper is to develop a theory of the
behaviours exhibited by sequential programs on execution that covers the
case where the programs have been split into fragments.
It is a fact of life that sequential programs are often fragmented.
We remark that an important reason for fragmentation of programs is that
the execution architecture at hand to execute them sets bounds to the
size of programs.
However, there may also be other reasons for program fragmentation, for
instance business economical reasons.

In~\cite{BL02a}, a start was made with a line of research in which
sequential programs that take the form of instruction sequences and the
behaviours exhibited by sequential programs on execution are
investigated (see e.g.~\cite{BBP05a,BM07e,PZ06a}).
In this line of research, the view is taken that the behaviour exhibited
by a sequential program on execution takes the form of a thread as
considered in basic thread algebra~\cite{BL02a}.%
\footnote
{In~\cite{BL02a}, basic thread algebra is introduced under the name
 basic polarized process algebra.
 Prompted by the development of thread algebra~\cite{BM04c}, which is a
 design on top of it, basic polarized process algebra has been renamed
 to basic thread algebra.
}
With the current paper, we carry on this line of research.
Therefore, we consider program fragment behaviours that take the form of
threads as considered in basic thread algebra.

We extend basic thread algebra with the barest mechanism for sequencing
of threads that are taken for program fragment behaviours.
This mechanism is called poly-threading.
Inherent in the behaviour exhibited by a program on execution is that it
does certain steps for the purpose of interacting with some service
provided by an execution environment.
In the setting of thread algebra, the use mechanism is introduced
in~\cite{BM04c} to allow for this kind of interaction.
Poly-threading supports the initialization of one of the services used
every time a thread is started up.
With poly-threading, a thread selection is made whenever a thread ends
up with the intent to achieve the start-up of another thread.
That thread selection can be made in two ways: by the terminating thread
or externally.
We show how thread selections of the latter kind can be internalized.

Both thread and service look to be special cases of a more general
notion of process.
Therefore, it is interesting to know how threads and services as
considered in the extension of basic thread algebra with poly-threading
relate to processes as considered in theories about concurrent processes
such as \ACP~\cite{BW90}, CCS~\cite{Mil89} and CSP~\cite{Hoa85}.
We show that threads and services as considered in the extension of
basic thread algebra with poly-threading can be viewed as processes that
are definable over the extension of \ACP\ with conditions introduced
in~\cite{BM05a}.

An analytic execution architecture is a model of a hypothetical
execution environment for sequential programs that is designed for the
purpose of explaining how a program may be executed.
The notion of analytic execution architecture defined in~\cite{BP04a} is
suited to sequential programs that have not been split into fragments.
We use the extension of basic thread algebra with poly-threading to
describe analytic execution architectures suited to sequential programs
that have been split into fragments.

In systems resulting from contemporary programming, we find distributed
multi-threading and threads that are program fragment behaviours.
For that reason, it is interesting to combine the theory of distributed
strategic interleaving developed in~\cite{BM07a} with the extension of
basic thread algebra with poly-threading.
We take up the combination by introducing two poly-threading covering
variations of the simplest form of interleaving for distributed
multi-threading considered in~\cite{BM07a}.

The line of research carried on in this paper has two main themes: the
theme of instruction sequences and the theme of threads.
Both~\cite{BM08b} and the current paper are concerned with program
fragmentation, but~\cite{BM08b} elaborates on the theme of instruction
sequences and the current paper elaborates on the theme of threads.
It happens that there are aspects of program fragmentation that can be
dealt with at the level of instruction sequences, but cannot be dealt
with at the level of threads.
In particular, the ability to replace special instructions in an
instruction sequence fragment by different ordinary instructions every
time execution is switched over to that fragment cannot be dealt with at
the level of threads.
Threads, which are intended for explaining the meaning of sequential
programs, turn out to be too abstract to deal with program fragmentation
in full.

This paper is organized as follows.
First, we review basic thread algebra and the use mechanism
(Sections~\ref{sect-BTA} and~\ref{sect-TSU}).
Next, we extend basic thread algebra with poly-threading and show how
external thread selections in poly-threading can be internalized
(Sections~\ref{sect-TApt} and~\ref{sect-internalization}).
Following this, we review \ACP\ with conditions and relate the extension
of basic thread algebra with poly-threading to \ACP\ with conditions
(Sections~\ref{sect-ACPc} and~\ref{sect-thr-serv-proc}).
Then, we discuss analytic execution architectures suited for programs
that have been fragmented (Section~\ref{sect-exe-arch}).
After that, we introduce forms of interleaving suited for
non-distributed and distributed multi-threading that cover
poly-threading (Sections~\ref{sect-TAptsi}, \ref{sect-TAptdsi}
and~\ref{sect-TAptdsifs}).
Finally, we make some concluding remarks (Section~\ref{sect-concl}).

Up to and including Section~\ref{sect-exe-arch}, this paper is a revision
of~\cite{BM08c}.
In that paper, the term ``sequential poly-threading'' stands for
``poly-threading in a setting where multi-threading or any other form of
concurrency is absent''.
We conclude in hindsight that the use of this term is unfortunate and do
not use it in the current paper.

\section{Basic Thread Algebra}
\label{sect-BTA}

In this section, we review \BTA, a form of process algebra which is
tailored to the description of the behaviour of deterministic sequential
programs under execution.
The behaviours concerned are called \emph{threads}.

In \BTA, it is assumed that there is a fixed but arbitrary finite set of
\emph{basic actions} $\BAct$ with $\Tau \not\in \BAct$.
We write $\BActTau$ for $\BAct \union \set{\Tau}$.
The members of $\BActTau$ are referred to as \emph{actions}.

The intuition is that each basic action performed by a thread is taken
as a command to be processed by a service provided by the execution
environment of the thread.
The processing of a command may involve a change of state of the service
concerned.
At completion of the processing of the command, the service produces a
reply value.
This reply is either $\True$ or $\False$ and is returned to the thread
concerned.

Although \BTA\ is one-sorted, we make this sort explicit.
The reason for this is that we will extend \BTA\ with additional sorts
in Sections~\ref{sect-TSU} and~\ref{sect-TApt}.

The algebraic theory \BTA\ has one sort: the sort $\Thr$ of
\emph{threads}.
To build terms of sort $\Thr$, \BTA\ has the following constants and
operators:
\begin{itemize}
\item
the \emph{deadlock} constant $\const{\DeadEnd}{\Thr}$;
\item
the \emph{termination} constant $\const{\Stop}{\Thr}$;
\item
for each $a \in \BActTau$, the \emph{postconditional composition}
operator $\funct{\pcc{\ph}{a}{\ph}}{\Thr \cprod \Thr}{\Thr}$.
\end{itemize}
Terms of sort $\Thr$ are built as usual (see e.g.~\cite{ST99a,Wir90a}).
Throughout the paper, we assume that there are infinitely many variables
of sort $\Thr$, including $x,y,z$.

We introduce \emph{action prefixing} as an abbreviation: $a \bapf p$,
where $p$ is a term of sort $\Thr$, abbreviates $\pcc{p}{a}{p}$.

Let $p$ and $q$ be closed terms of sort $\Thr$ and $a \in \BActTau$.
Then $\pcc{p}{a}{q}$ will perform action $a$, and after that proceed as
$p$ if the processing of $a$ leads to the reply $\True$ (called a
positive reply), and proceed as $q$ if the processing of $a$ leads to
the reply $\False$ (called a negative reply).
The action $\Tau$ plays a special role.
It is a concrete internal action: performing $\Tau$ will never lead to a
state change and always lead to a positive reply, but notwithstanding
all that its presence matters.

\BTA\ has only one axiom.
This axiom is given in Table~\ref{axioms-BTA}.%
\begin{table}[!t]
\caption{Axiom of \BTA}
\begin{eqntbl}
\begin{axcol}
\pcc{x}{\Tau}{y} = \pcc{x}{\Tau}{x}                      & \axiom{T1}
\end{axcol}
\end{eqntbl}
\label{axioms-BTA}
\end{table}
Using the abbreviation introduced above, axiom T1 can be written as
follows: $\pcc{x}{\Tau}{y} = \Tau \bapf x$.

Each closed \BTA\ term of sort $\Thr$ denotes a finite thread, i.e.\ a
thread of which the length of the sequences of actions that it can
perform is bounded.
Guarded recursive specifications give rise to infinite threads.

A \emph{guarded recursive specification} over \BTA\ is a set of
recursion equations $E = \set{X = t_X \where X \in V}$, where $V$ is a
set of variables of sort $\Thr$ and each $t_X$ is a term of the form
$\DeadEnd$, $\Stop$ or $\pcc{t}{a}{t'}$ with $t$ and $t'$ \BTA\ terms of
sort $\Thr$ that contain only variables from $V$.
We write $\vars(E)$ for the set of all variables that occur on the
left-hand side of an equation in $E$.
We are only interested in models of \BTA\ in which guarded recursive
specifications have unique solutions, such as the projective limit model
of \BTA\ presented in~\cite{BB03a}.
A thread that is the solution of a finite guarded recursive
specification over \BTA\ is called a \emph{finite-state} thread.

We extend \BTA\ with guarded recursion by adding constants for solutions
of guarded recursive specifications and axioms concerning these
additional constants.
For each guarded recursive specification $E$ and each $X \in \vars(E)$,
we add a constant of sort $\Thr$ standing for the unique solution of $E$
for $X$ to the constants of \BTA.
The constant standing for the unique solution of $E$ for $X$ is denoted
by $\rec{X}{E}$.
Moreover, we add the axioms for guarded recursion given in
Table~\ref{axioms-REC} to \BTA,%
\begin{table}[!t]
\caption{Axioms for guarded recursion}
\begin{eqntbl}
\begin{saxcol}
\rec{X}{E} = \rec{t_X}{E} & \mif X \!=\! t_X \in E       & \axiom{RDP}
\\
E \Implies X = \rec{X}{E} & \mif X \in \vars(E)          & \axiom{RSP}
\end{saxcol}
\end{eqntbl}
\label{axioms-REC}
\end{table}
where we write $\rec{t_X}{E}$ for $t_X$ with, for all $Y \in \vars(E)$,
all occurrences of $Y$ in $t_X$ replaced by $\rec{Y}{E}$.
In this table, $X$, $t_X$ and $E$ stand for an arbitrary variable of
sort $\Thr$, an arbitrary \BTA\ term of sort $\Thr$ and an arbitrary
guarded recursive specification over \BTA, respectively.
Side conditions are added to restrict the variables, terms and guarded
recursive specifications for which $X$, $t_X$ and $E$ stand.
%

We will write \BTA+\REC\ for \BTA\ extended with the constants for
solutions of guarded recursive specifications and axioms RDP and RSP.

\section{Interaction of Threads with Services}
\label{sect-TSU}

A thread may perform certain basic actions only for the sake of having
itself affected by some service.
When processing a basic action performed by a thread, a service affects
that thread by returning a reply value to the thread at completion of
the processing of the basic action.
In this section, we introduce the use mechanism, which is concerned with
this kind of interaction between threads and services.%
\footnote
{This version of the use mechanism was first introduced in~\cite{BM04c}.
 In later papers, it is also called thread-service composition.}

It is assumed that there is a fixed but arbitrary finite set $\Foci$ of
\emph{foci} and a fixed but arbitrary finite set $\Meth$ of
\emph{methods}.
Each focus plays the role of a name of a service provided by the
execution environment that can be requested to process a command.
Each method plays the role of a command proper.
For the set $\BAct$ of basic actions, we take the set
$\set{f.m \where f \in \Foci, m \in \Meth}$.
A thread performing a basic action $f.m$ is considered to make a request
to a service that is known to the thread under the name $f$ to process
command~$m$.

We introduce yet another sort: the sort $\Serv$ of \emph{services}.
However, we will not introduce constants and operators to build terms
of this sort.
$\Serv$ is considered to stand for the set of all services.
We identify services with functions
$\funct{H}{\neseqof{\Meth}}{\set{\True,\False,\Blocked}}$
that satisfy the following condition:
\begin{ldispl}
\Forall{\rho \in \neseqof{\Meth},m \in \Meth}
 {(H(\rho) = \Blocked \Implies H(\rho \concat \seq{m}) =
   \Blocked)}\;.\footnotemark
\end{ldispl}%
\footnotetext
{We write $\seqof{D}$ for the set of all finite sequences with elements
 from set $D$ and $\neseqof{D}$ for the set of all non-empty finite
 sequences with elements from set $D$.
 We use the following notation for finite sequences:
 $\emptyseq$ for the empty sequence,
 $\seq{d}$ for the sequence having $d$ as sole element,
 $\sigma \concat \sigma'$ for the concatenation of finite sequences
 $\sigma$ and $\sigma'$, and
 $\len(\sigma)$ for the length of finite sequence $\sigma$.}%
We write $\Services$ for the set of all services and $\Replies$ for
the set $\set{\True,\False,\Blocked}$.
Given a service $H$ and a method $m \in \Meth$,
the \emph{derived service} of $H$ after processing $m$,
written $\derive{m}H$, is defined by
$\derive{m}H(\rho) = H(\seq{m} \concat \rho)$.

A service $H$ can be understood as follows:
\begin{itemize}
\item
if $H(\seq{m}) \neq \Blocked$, then the request to process $m$ is
accepted by the service, the reply is $H(\seq{m})$, and the service
proceeds as $\derive{m}H$;
\item
if $H(\seq{m}) = \Blocked$, then the request to process $m$ is not
accepted by the service.
\end{itemize}

For each $f \in \Foci$, we introduce the \emph{use} operator
$\funct{\use{\ph}{f}{\ph}}{\Thr \cprod \Serv}{\Thr}$.
Intuitively, $\use{p}{f}{H}$ is the thread that results from processing
all basic actions performed by thread $p$ that are of the form $f.m$ by
service $H$.
When a basic action of the form $f.m$ performed by thread $p$ is
processed by service $H$, it is turned into the internal action $\Tau$
and postconditional composition is removed in favour of action prefixing
on the basis of the reply value produced.

The axioms for the use operators are given in Table~\ref{axioms-tsc}.%
\begin{table}[!t]
\caption{Axioms for use}
\begin{eqntbl}
\begin{saxcol}
\use{\Stop}{f}{H} = \Stop                            & & \axiom{TSU1} \\
\use{\DeadEnd}{f}{H} = \DeadEnd                      & & \axiom{TSU2} \\
\use{\Tau \bapf x}{f}{H} =
                          \Tau \bapf (\use{x}{f}{H}) & & \axiom{TSU3} \\
\use{(\pcc{x}{g.m}{y})}{f}{H} =
\pcc{(\use{x}{f}{H})}{g.m}{(\use{y}{f}{H})}
 & \mif \lnot f = g                                     & \axiom{TSU4} \\
\use{(\pcc{x}{f.m}{y})}{f}{H} =
\Tau \bapf (\use{x}{f}{\derive{m}H})
 & \mif H(\seq{m}) = \True                             & \axiom{TSU5} \\
\use{(\pcc{x}{f.m}{y})}{f}{H} =
\Tau \bapf (\use{y}{f}{\derive{m}H})
 & \mif H(\seq{m}) = \False                            & \axiom{TSU6} \\
\use{(\pcc{x}{f.m}{y})}{f}{H} = \DeadEnd
 & \mif H(\seq{m}) = \Blocked                          & \axiom{TSU7}
\end{saxcol}
\end{eqntbl}
\label{axioms-tsc}
\end{table}
In this table, $f$ and $g$ stand for arbitrary foci from $\Foci$ and $m$
stands for an arbitrary method from~$\Meth$.
Axioms TSU3 and TSU4 express that the action $\Tau$ and basic actions of
the form $g.m$ with $f \neq g$ are not processed.
Axioms TSU5 and TSU6 express that a thread is affected by a service as
described above when a basic action of the form $f.m$ performed by the
thread is processed by the service.
Axiom TSU7 expresses that deadlock takes place when a basic action to be
processed is not accepted.

Let $T$ stand for either \BTA\ or \BTA+\REC.
Then we will write $T$+\TSU\ for $T$, taking the set
$\set{f.m \where f \in \Foci, m \in \Meth}$ for $\BAct$, extended with
the use operators and the axioms from Table~\ref{axioms-tsc}.

\section{Poly-Threading}
\label{sect-TApt}

\BTA\ is a theory of the behaviours exhibited by sequential programs on
execution.
To cover the case where the programs have been split into fragments, we
extend \BTA\ in this section with the barest mechanism for sequencing of
threads that are taken for fragments.
The resulting theory is called \TApt.

Our general view on the way of achieving a joint behaviour of the
program fragments in a collection of program fragments between which
execution can be switched is as follows:
\begin{itemize}
\item
there can only be a single program fragment being executed at any stage;
\item
the program fragment in question may make any program fragment in the
collection the one being executed;
\item
making another program fragment the one being executed is effected by
executing a special instruction for switching over execution;
\item
any program fragment can be taken for the one being executed initially.
\end{itemize}
In order to obtain such a joint behaviour from the behaviours of the
program fragments on execution, a mechanism is needed by which the
start-up of another program fragment behaviour is effectuated whenever a
program fragment behaviour ends up with the intent to achieve such a
start-up.
In the setting of \BTA, taking threads for program fragment behaviours,
this requires the introduction of an additional sort, additional
constants and additional operators.
In doing so it is supposed that a collection of threads that corresponds
to a collection of program fragments between which execution can be
switched takes the form of a sequence, called a thread vector.

Like in \BTA+\TSU, it is assumed that there is a fixed but arbitrary
finite set $\Foci$ of foci and a fixed but arbitrary finite set $\Meth$
of methods.
It is also assumed that $\tls \in \Foci$ and $\init \in \Meth$.
The focus $\tls$ and the method $\init$ play special roles: $\tls$ is
the focus of a service that is initialized each time a thread is started
up by the mechanism referred to above and $\init$ is the initialization
method of that service.
For the set $\BAct$ of basic actions, we take again the set
$\set{f.m \where f \in \Foci, m \in \Meth}$.

\TApt\ has the sort $\Thr$ of \BTA\ and in addition the sort $\TV$ of
\emph{thread vectors}.
To build terms of sort $\Thr$, \TApt\ has the constants and operators
of \BTA\ and in addition the following additional constants and
operators:
\begin{itemize}
\item
for each $i \in \Nat$,
the \emph{internally controlled switch-over} constant
$\const{\Switch{i}}{\Thr}$;
\item
the \emph{externally controlled switch-over} constant
$\const{\Extern}{\Thr}$;
\item
the \emph{poly-threading} operator
$\funct{\sptop}{\Thr \cprod \TV}{\Thr}$;
\item
for each $k \in \Nat^+$,%
\footnote
{We write $\Nat^+$ for the set $\set{n \in \Nat \where n > 0}$.
 Throughout the paper, we use the convention that $k$ and $n$ stand for
 arbitrary elements of $\Nat^+$ and $\Nat$, respectively.}
the $k$-ary \emph{external choice} operator
$\funct{\choiceop{k}}
  {\underbrace{\Thr \cprod \cdots \cprod \Thr}_{k \;\mathrm{times}}}
  {\Thr}$.
\end{itemize}
To build terms of sort $\TV$, \TApt\ has the following constants and
operators:
\begin{itemize}
\item
the \emph{empty thread vector} constant $\const{\emptyseq}{\TV}$;
\item
the \emph{singleton thread vector} operator
$\funct{\seq{\ph}}{\Thr}{\TV}$;
\item
the \emph{thread vector concatenation} operator
$\funct{\ph \concat \ph}{\TV \cprod \TV}{\TV}$.
\end{itemize}
Throughout the paper, we assume that there are infinitely many variables
of sort $\TV$, including $\alpha,\beta,\gamma$.

In the context of the poly-threading operator $\sptop$, the constants
$\Switch{i}$ and $\Extern$ are alternatives for the constant $\Stop$
which produce additional effects.
Let $p$, $p_1$, \ldots, $p_n$ be closed terms of sort $\Thr$.
Then $\spt{p}{\seq{p_1} \concat \ldots \concat \seq{p_n}}$ first behaves as
$p$, but when $p$ terminates:
\begin{itemize}
\item
in the case where $p$ terminates with $\Stop$, it terminates;
\item
in the case where $p$ terminates with $\Switch{i}$:
\begin{itemize}
\item
it continues by behaving as
$\spt{p_i}{\seq{p_1} \concat \ldots \concat \seq{p_n}}$
if $1 \leq i \leq n$,
\item
it deadlocks otherwise;
\end{itemize}
\item
in the case where $p$ terminates with $\Extern$, it continues by
behaving as one of
$\spt{p_1}{\seq{p_1} \concat \ldots \concat \seq{p_n}}$, \ldots,
$\spt{p_n}{\seq{p_1} \concat \ldots \concat \seq{p_n}}$
or it deadlocks.
\end{itemize}
Moreover, the basic action $\tls.\init$ is performed between termination
and continuation.
In the case where $p$ terminates with $\Extern$, the choice between the
alternatives is made externally.
Nothing is stipulated about the effect that the constants $\Switch{i}$
and $\Extern$ produce in the case where they occur outside the context
of the poly-threading operator.

The poly-threading operator concerns sequencing of threads.
A thread selection involved in sequencing of threads is called an
\emph{autonomous thread selection} if the selection is made by the
terminating thread.
Otherwise, it is called a \emph{non-autonomous thread selection}.
The constants $\Switch{i}$ are meant for autonomous thread selections
and the constant $\Extern$ is meant for non-autonomous thread
selections.
We remark that non-autonomous thread selections are immaterial to the
joint behaviours of program fragments referred to above.

In the case of a non-autonomous thread selection, it comes to an
external choice between a number of threads.
The external choice operator $\choiceop{k}$ concerns external choice
between $k$ threads.
Let $p_1$, \ldots, $p_k$ be closed terms of sort $\Thr$.
Then $\choice{k}{p_1,\ldots,p_k}$ behaves as the outcome of an external
choice between $p_1$, \ldots, $p_k$ and $\DeadEnd$.

\TApt\ has the axioms of \BTA\ and in addition the axioms given in
Table~\ref{axioms-spt}.%
%
\begin{table}[!t]
\caption{Axioms for poly-threading}
\begin{eqntbl}
\begin{saxcol}
\spt{\Stop}{\alpha} = \Stop                          & & \axiom{SPT1} \\
\spt{\DeadEnd}{\alpha} = \DeadEnd                    & & \axiom{SPT2} \\
\spt{\pcc{x}{a}{y}}{\alpha} =
\pcc{\spt{x}{\alpha}}{a}{\spt{y}{\alpha}}            & & \axiom{SPT3} \\
\spt{\Switch{i}}{\seq{x_1} \concat \ldots \concat \seq{x_n}} =
\tls.\init \bapf \spt{x_i}{\seq{x_1} \concat \ldots \concat \seq{x_n}}
                               & \mif 1 \leq i \leq n  & \axiom{SPT4} \\
\spt{\Switch{i}}{\seq{x_1} \concat \ldots \concat \seq{x_n}} =
\DeadEnd                       & \mif i = 0 \lor i > n & \axiom{SPT5} \\
\spt{\Extern}{\seq{x_1} \concat \ldots \concat \seq{x_k}} =
{} \\
\multicolumn{2}{@{}l@{\;\;}}
{\quad
 \choice{k}
  {\tls.\init \bapf
   \spt{x_1}{\seq{x_1} \concat \ldots \concat \seq{x_k}},\ldots,
   \tls.\init \bapf
   \spt{x_k}{\seq{x_1} \concat \ldots \concat \seq{x_k}}}}
                                                       & \axiom{SPT6} \\
\spt{\Extern}{\emptyseq} = \DeadEnd                  & & \axiom{SPT7}
\end{saxcol}
\end{eqntbl}
\label{axioms-spt}
\end{table}
In this table, $a$ stands for an arbitrary action from $\BActTau$.
The additional axioms express that threads are sequenced by
poly-threading as described above.
There are no axioms for the external choice operators because their basic
properties cannot be expressed as equations or conditional equations.
For each $k \in \Nat^+$, the basic properties of $\choiceop{k}$ are
expressed by the following disjunction of equations:
$\OR{i \in [1,k]} \choice{k}{x_1,\ldots,x_k} = x_i \lor
 \choice{k}{x_1,\ldots,x_k} = \DeadEnd$.%
\footnote
{We use the notation $[n,m]$, where $n,m \in \Nat$, for the set
 $\set{i \in \Nat \where n \leq i \leq m}$.}

To be fully precise, we should give axioms concerning the constants and
operators to build terms of the sort $\TV$ as well.
We refrain from doing so because the constants and operators concerned
are the usual ones for sequences.
Similar remarks apply to the sort $\DTV$ introduced later and will not
be repeated.

Guarded recursion can be added to \TApt\ as it is added to \BTA\ in
Section~\ref{sect-BTA}.
We will write \TApt+\REC\ for \TApt\ extended with the constants for
solutions of guarded recursive specifications and axioms \RDP\ and \RSP.

The use mechanism can be added to \TApt\ as it is added to \BTA\ in
Section~\ref{sect-TSU}.
Let $T$ stand for either \TApt\ or \TApt+\REC.
Then we will write $T$+\TSU\ for $T$ extended with
the use operators and the axioms from Table~\ref{axioms-tsc}.

\section{Internalization of Non-Autonomous Thread Selection}
\label{sect-internalization}

In the case of non-autonomous thread selection, the selection of a
thread is made externally.
In this section, we show how non-autonomous thread selection can be
internalized.
For that purpose, we first extend \TApt\ with postconditional
switching.
Postconditional switching is like postconditional composition, but
covers the case where services processing basic actions produce reply
values from the set $\Nat$ instead of reply values from the set
$\set{\True,\False}$.
Postconditional switching is convenient when internalizing
non-autonomous thread selection, but it is not necessary.

For each $a \in \BActTau$ and $k \in \Nat^+$, we introduce
the $k$-ary \emph{postconditional switch} operator
$\funct{\pcsop{k}{a}}
  {\underbrace{\Thr \cprod \cdots \cprod \Thr}_{k \;\mathrm{times}}}
  {\Thr}$.
Let $p_1$, \ldots, $p_k$ be closed terms of sort $\Thr$.
Then $\pcs{k}{a}{p_1,\ldots,p_k}$ will first perform action $a$, and
then proceed as $p_1$ if the processing of $a$ leads to the reply~$1$,
\ldots, $p_k$ if the processing of $a$ leads to the reply $k$.

The axioms for the postconditional switching operators are given in
Table~\ref{axioms-npcc}.%
\begin{table}[!t]
\caption{Axioms for postconditional switching}
\begin{eqntbl}
\begin{axcol}
\pcs{k}{\Tau}{x_1,\ldots,x_k} = \pcs{k}{\Tau}{x_1,\ldots,x_1}       & \\
\spt{\pcs{k}{a}{x_1,\ldots,x_k}}{\alpha} =
\pcs{k}{a}{\spt{x_1}{\alpha},\ldots,\spt{x_k}{\alpha}}              & \\
\use{\pcs{k}{\Tau}{x_1,\ldots,x_k}}{f}{H} =
\pcs{k}{\Tau}{\use{x_1}{f}{H},\ldots,\use{x_k}{f}{H}}               & \\
\use{\pcs{k}{g.m}{x_1,\ldots,x_k}}{f}{H} =
\pcs{k}{g.m}{\use{x_1}{f}{H},\ldots,\use{x_k}{f}{H}}
                               & \mif \lnot f = g                      \\
\use{\pcs{k}{f.m}{x_1,\ldots,x_k}}{f}{H} =
\Tau \bapf (\use{x_i}{f}{\derive{m}H})
                               & \mif H(\seq{m}) = i \land i \in [1,k] \\
\use{\pcs{k}{f.m}{x_1,\ldots,x_k}}{f}{H} = \DeadEnd
                               & \mif \lnot H(\seq{m}) \in [1,k]
\end{axcol}
\end{eqntbl}
\label{axioms-npcc}
\end{table}
In this table, $a$ stands for an arbitrary action from $\BActTau$, $f$
and $g$ stand for arbitrary foci from $\Foci$, and $m$ stands for an
arbitrary method from~$\Meth$.

We proceed with the internalization of non-autonomous thread selections.
Let $p$, $p_1$, \ldots, $p_k$ be closed terms of sort $\Thr$.
The  idea is that
$\spt{p}{\seq{p_1} \concat \ldots \concat \seq{p_k}}$
can be internalized by:
\begin{itemize}
\item
replacing in $\spt{p}{\seq{p_1} \concat \ldots \concat \seq{p_k}}$ all
occurrences of $\Extern$ by $\Switch{k{+}1}$;
\item
appending a thread that can make the thread selections to the thread
vector.
\end{itemize}
Simultaneous with the replacement of all occurrences of $\Extern$ by
$\Switch{k{+}1}$, all occurrences of $\Switch{k{+}1}$ must be replaced
by $\DeadEnd$ to prevent inadvertent selections of the appended thread.
When making a thread selection, the appended thread has to request the
external environment to give the position of the thread that it would
have selected itself.
We make the simplifying assumption that the external environment can be
viewed as a service.

Let $p$, $p_1$, \ldots, $p_k$ be closed terms of sort $\Thr$.
\sloppy
Then the \emph{internalization} of
$\spt{p}{\seq{p_1} \concat \ldots \concat \seq{p_k}}$ is
\begin{ldispl}
\spt{\rho(p)}
 {\seq{\rho(p_1)} \concat \ldots \concat
  \seq{\rho(p_k)} \concat
  \seq{\pcs{k}{\ext.\sel}{\Switch{1},\ldots,\Switch{k}}}}\;,
\end{ldispl}%
where $\rho(p')$ is $p'$ with simultaneously all occurrences of
$\Extern$ replaced by $\Switch{k{+}1}$ and all occurrences of
$\Switch{k{+}1}$ replaced by $\DeadEnd$.
Here, it is assumed that $\ext \in \Foci$ and $\sel \in \Meth$.

Postconditional switching is not really necessary for internalization.
Let $k_1 = \lfloor k/2 \rfloor$, $k_2 = \lfloor k_1/2 \rfloor$,
$k_3 = \lfloor (k - k_1)/2 \rfloor$, \ldots\ .
Using postconditional composition, first a selection can be made
between $\set{p_1,\ldots,p_{k_1}}$ and $\set{p_{k_1+1},\ldots,p_k}$,
next a selection can be made between $\set{p_1,\ldots,p_{k_2}}$ and
$\set{p_{k_2+1},\ldots,p_{k_1}}$ or between
$\set{p_{k_1+1},\ldots,p_{k_3}}$ and $\set{p_{k_3+1},\ldots,p_k}$,
depending on the outcome of the previous selection, etcetera.
In this way, the number of actions performed to select a thread is
between $\lfloor ^2\log(k) \rfloor$ and $\lceil ^2\log(k) \rceil$.

\section{\ACP\ with Conditions}
\label{sect-ACPc}

In Section~\ref{sect-thr-serv-proc}, we will investigate the
connections of threads and services with the processes considered in
\ACP-style process algebras.
We will focus on \ACPc, an extension of \ACP\ with conditions introduced
in~\cite{BM05a}.
In this section, we shortly review \ACPc.

\ACPc\ is an extension of \ACP\ with conditional expressions in which
the conditions are taken from a Boolean algebra.
\ACPc\ has two sorts:
(i)  the sort $\Proc$ of \emph{processes},
(ii) the sort $\Cond$ of \emph{conditions}.
In \ACPc, it is assumed that the following has been given:
a fixed but arbitrary set $\Act$ (of actions), with
$\dead \not\in \Act$,
a fixed but arbitrary set $\ACond$ (of atomic conditions), and
a fixed but arbitrary commutative and associative function
$\funct{\commm}{\Act \union \set{\dead} \cprod \Act \union \set{\dead}}
       {\Act \union \set{\dead}}$ such that
$\dead \commm a = \dead$ for all $a \in \Act \union \set{\dead}$.
The function $\commm$ is regarded to give the result of synchronously
performing any two actions for which this is possible, and to be $\dead$
otherwise.
Henceforth, we write $\Actd$ for $\Act \union \set{\dead}$.

Let $p$ and $q$ be closed terms of sort $\Proc$, $\zeta$ and $\xi$
be closed term of sort $\Cond$, $a \in \Act$, $H \subseteq \Act$, and
$\eta \in \ACond$.
Intuitively, the constants and operators to build terms of sort $\Proc$
that will be used to define the processes to which threads and services
correspond can be explained as follows:
\begin{itemize}
\item
$\dead$ can neither perform an action nor terminate successfully;
\item
$a$ first performs action $a$ unconditionally and then terminates
successfully;
\item
$p \altc q$ behaves either as $p$ or as $q$, but not both;
\item
$p \seqc q$ first behaves as $p$, but when $p$ terminates successfully
it continues as $q$;
\item
$\zeta \gc p$ behaves as $p$ under condition $\zeta$;
\item
$p \parc q$ behaves as the process that proceeds with $p$ and $q$ in
parallel;
\item
$\encap{H}(p)$ behaves the same as $p$, except that actions from $H$ are
blocked.
\end{itemize}
Intuitively, the constants and operators to build terms of sort $\Cond$
that will be used to define the processes to which threads and services
correspond can be explained as follows:
\begin{itemize}
\item
$\eta$ is an atomic condition;
\item
$\bot$ is a condition that never holds;
\item
$\top$ is a condition that always holds;
\item
$\bcompl \zeta$ is the opposite of $\zeta$;
\item
$\zeta \join \xi$ is either $\zeta$ or $\xi$;
\item
$\zeta \meet \xi$ is both $\zeta$ and $\xi$.
\end{itemize}
The remaining operators of \ACPc\ are of an auxiliary nature.
They are needed to axiomatize \ACPc.
The axioms of \ACPc\ are given in~\cite{BM05a}.

We write $\vAltc{i \in \mathcal{I}} p_i$,
where $\mathcal{I} = \set{i_1,\ldots,i_n}$
and $p_{i_1},\ldots,p_{i_n}$ are terms of sort $\Proc$,
for $p_{i_1} \altc \ldots \altc p_{i_n}$.
The convention is that $\vAltc{i \in \mathcal{I}} p_i$ stands for
$\dead$ if $\mathcal{I} = \emptyset$.
We use the notation $\cond{p}{\zeta}{q}$,
where $p$ and $q$ are terms of sort $\Proc$
and $\zeta$ is a term of sort $\Cond$,
for $\zeta \gc p \altc \bcompl{\zeta} \gc q$.

A process is considered definable over \ACPc\ if there exists a guarded
recursive specification over \ACPc\ that has that process as its
solution.

A \emph{recursive specification} over \ACPc\ is a set of recursion
equations $E = \set{X = t_X \where\linebreak[2] X \in V}$, where $V$ is
a set of variables and each $t_X$ is a term of sort $\Proc$ that only
contains variables from $V$.
Let $t$ be a term of sort $\Proc$ containing a variable $X$.
Then an occurrence of $X$ in $t$ is \emph{guarded} if $t$ has a subterm
of the form $a \seqc t'$ where $a \in \Act$ and $t'$ is a term
containing this occurrence of $X$.
Let $E$ be a recursive specification over \ACPc.
Then $E$ is a \emph{guarded recursive specification} if, in each
equation $X = t_X \in E$, all occurrences of variables in $t_X$ are
guarded or $t_X$ can be rewritten to such a term using the axioms of
\ACPc\ in either direction and/or the equations in $E$ except the
equation $X = t_X$ from left to right.
We only consider models of \ACPc\ in which guarded recursive
specifications have unique solutions, such as the full splitting
bisimulation models of \ACPc\ presented in~\cite{BM05a}.

For each guarded recursive specification $E$ and each variable
$X$ that occurs as the left-hand side of an equation in $E$, we
introduce a constant of sort $\Proc$ standing for the unique solution of
$E$ for $X$.
This constant is denoted by $\rec{X}{E}$.
The axioms for guarded recursion are also given in~\cite{BM05a}.

In order to express the use operators, we need an extension of \ACPc\
with action renaming operators.
Intuitively, the action renaming operator $\aren{f}$, where
$\funct{f}{\Act}{\Act}$, can be explained as follows: $\aren{f}(p)$
behaves as $p$ with each action replaced according to $f$.
The axioms for action renaming are the ones given in~\cite{Fok00} and in
addition the equation $\rho_f(\phi \gc x) = \phi \gc \rho_f(x)$.
We write $\aren{a' \mapsto a''}$ for the renaming operator $\aren{g}$
with $g$ defined by $g(a') = a''$ and $g(a) = a$ if $a \neq a'$.

In order to explain the connection of threads and services with \ACPc\
fully, we need an extension of \ACPc\ with the condition evaluation
operators $\ceval{h}$ introduced in~\cite{BM05a}.
Intuitively, the condition evaluation operator $\ceval{h}$, where $h$ is
a function on conditions that is preserved by $\bot$, $\top$, $\bcompl$,
$\join$ and $\meet$, can be explained as follows: $\ceval{h}(p)$ behaves
as $p$ with each condition replaced according to $h$.
The important point is that, if $h(\zeta) \in \set{\bot,\top}$,
all subterms of the form $\zeta \gc q$ can be eliminated.
The axioms for condition evaluation are also given in~\cite{BM05a}.

\section{Threads, Services and \ACPc-Definable Processes}
\label{sect-thr-serv-proc}

In this section, we relate threads and services as considered in
\TApt+\REC+\TSU\ to processes that are definable over \ACPc\ with
action renaming.

For that purpose, $\Act$, $\commm$ and $\ACond$ are taken as follows:
\begin{ldispl}
\begin{aeqns}
\Act  & = &
\set{\snd_f(d) \where f \in \Foci, d \in \Meth \union \Replies} \union
\set{\rcv_f(d) \where f \in \Foci, d \in \Meth \union \Replies}
\\ & {} \union {} &
\set{\snd_\ext(n) \where n \in \Nat} \union
\set{\rcv_\ext(n) \where n \in \Nat} \union
\set{\stp,\ol{\stp},\stp^*,\iact}
\\ & {} \union {} &
\set{\snd_\serv(r) \where r \in \Replies} \union
\set{\rcv_\serv(m) \where m \in \Meth}\;;
\end{aeqns}
\end{ldispl}%
for all $a \in \Act$, $f \in \Foci$, $d \in \Meth \union \Replies$,
$m \in \Meth$, $r \in \Replies$ and $n \in \Nat$:
\begin{ldispl}
\begin{aeqns}
\snd_f(d) \commm \rcv_f(d) = \iact \;,
\\
\snd_f(d) \commm a = \dead & & \mif a \neq \rcv_f(d)\;,
\\
a \commm \rcv_f(d) = \dead & & \mif a \neq \snd_f(d)\;,
\bigeqnsep
\snd_\ext(n) \commm \rcv_\ext(n) = \iact \;,
\\
\snd_\ext(n) \commm a = \dead & & \mif a \neq \rcv_\ext(n)\;,
\\
a \commm \rcv_\ext(n) = \dead & & \mif a \neq \snd_\ext(n)\;,
\end{aeqns}
\qquad\;
\begin{aeqns}
\stp \commm \ol{\stp} = \stp^*\;,
\\
\stp \commm a = \dead      & & \mif a \neq \ol{\stp}\;,
\\
a \commm \ol{\stp} = \dead & & \mif a \neq \stp\;,
\eqnsep
\iact \commm a = \dead\;,
\eqnsep
\snd_\serv(r) \commm a = \dead\;,
\\
a \commm \rcv_\serv(m) = \dead\;;
\end{aeqns}
\end{ldispl}%
and
\begin{ldispl}
\begin{aeqns}
\ACond & = &
\set{H(\seq{m}) = r \where
     H \in \Services, m \in \Meth, r \in \Replies}\;.
\end{aeqns}
\end{ldispl}%

For each $f \in \Foci$, the set $A_f \subseteq \Act$ and the function
$\funct{R_f}{\Act}{\Act}$ are defined as follows:
\begin{ldispl}
A_f =
\set{\snd_f(d) \where d \in \Meth \union \Replies} \union
\set{\rcv_f(d) \where d \in \Meth \union \Replies}\;;
\end{ldispl}%
for all $a \in \Act$, $m \in \Meth$ and $r \in \Replies$:
\begin{ldispl}
\begin{aeqns}
R_f(\snd_\serv(r)) = \snd_f(r)\;, \\
R_f(\rcv_\serv(m)) = \rcv_f(m)\;, \\
R_f(a)             = a
 & & \mif \AND{r' \in \Replies}{a \neq \snd_\serv(r')} \land
          \AND{m' \in \Meth}{a \neq \rcv_\serv(m')}\;.
\end{aeqns}
\end{ldispl}%
The sets $A_f$ and the functions $R_f$ are used below to express the use
operators in terms of the operators of \ACPc\ with action renaming.

For convenience, we introduce a special notation.
Let $\alpha$ be a term of sort $\TV$,
let $p_1,\ldots,p_n$ be terms of sort $\Thr$ such that
$\alpha = \seq{p_1} \concat \ldots \concat \seq{p_n}$, and
let $i \in [1,n]$.
Then we write $\alpha[i]$ for $p_i$.

We proceed with relating threads and services as considered in
\TApt+\REC+\linebreak[2]\TSU\ to processes definable over \ACPc\ with
action renaming.
The underlying idea is that threads and services can be viewed as
processes that are definable over \ACPc\ with action renaming.
We define those processes by means of a translation function $\Mt{\ph}$
from the set of all terms of sort $\Thr$ to the set of all function from
the set of all terms of sort $\TV$ to the set of all terms
of sort $\Proc$ and a translation function $\Ms{\ph}$ from the set of
all services to the set of all terms of sort $\Proc$.
These translation functions are defined inductively by the equations
given in Table~\ref{def-translation},%
\begin{table}[!t]
\caption{Definition of translation functions}
\begin{eqntbl}
\begin{seqncol}
\Mt{X}(\alpha) = X
\\
\Mt{\Stop}(\alpha) = \stp
\\
\Mt{\DeadEnd}(\alpha) = \iact \seqc \dead
\\
\Mt{\pcc{t_1}{\Tau}{t_2}}(\alpha) =
\iact \seqc \iact \seqc \Mt{t_1}(\alpha)
\\
\Mt{\pcc{t_1}{f.m}{t_2}}(\alpha) =
\snd_f(m) \seqc
(\rcv_f(\True) \seqc \Mt{t_1}(\alpha) \altc
 \rcv_f(\False) \seqc \Mt{t_2}(\alpha))
\\
\Mt{\Switch{i}}(\alpha) = \tls.\init \seqc \Mt{\alpha[i]}(\alpha)
 & \mif 1 \leq i \leq \len(\alpha)
\\
\Mt{\Switch{i}}(\alpha) = \iact \seqc \dead
 & \mif i = 0 \lor i > \len(\alpha)
\\
\Mt{\Extern}(\alpha) =
\vAltc{i \in [1,\len(\alpha)]}
 \rcv_\ext(i) \seqc \tls.\init \seqc \Mt{\alpha[i]}(\alpha) \altc
 \iact \seqc \dead
\\
\Mt{\spt{t}{\alpha'}}(\alpha) = \Mt{t}(\alpha')
\\
\Mt{\choice{k}{t_1,\ldots,t_k}}(\alpha) =
\vAltc{i \in [1,k]} \rcv_\ext(i) \seqc \Mt{t_i}(\alpha) \altc
\iact \seqc \dead
\\
\Mt{\rec{X}{E}}(\alpha) =
\rec{X}{\set{X = \Mt{t}(\alpha) \where X = t \,\in\, E}}
\\
\Mt{\use{t}{f}{H}}(\alpha) =
\aren{\stp^* \mapsto \stp}
 (\encap{\set{\stp,\ol{\stp}}}
   (\encap{A_f}(\Mt{t}(\alpha) \parc \aren{R_f}(\Ms{H}))))
\eqnsep
\Ms{H} = \rec{X_H}{\set{X_{H'} = t_{H'} \where H' \in \Services}}
\end{seqncol}
\end{eqntbl}
\label{def-translation}
\end{table}
where we write in the last equation $t_{H'}$ for the term
\begin{ldispl}
{} \hspace*{-1em}
\Altc{m \in \Meth}
  \rcv_\serv(m) \seqc \snd_\serv({H'}(\seq{m})) \seqc
  (\cond{X_{\derive{m} {H'}}}
        {{H'}(\seq{m}) \!=\! \True \join {H'}(\seq{m}) \!=\! \False}
        {X_{H'}})
\\ {} \hspace*{-1em} \altc
  \ol{\stp}\;.
\end{ldispl}%
Let $p$ be a closed term of sort $\Thr$.
Then the \emph{process algebraic interpretation} of $p$ is
$\Mt{p}(\emptyseq)$.
Henceforth, we write $\Mt{p}$ for $\Mt{p}(\emptyseq)$.

Notice that \ACP\ is sufficient for the translation of terms of sort
$\Thr$: no conditional expressions occur in the translations.
For the translation of services, we need the full power of \ACPc.

The translations given above preserve the axioms of \TApt+\REC+\TSU.
Roughly speaking, this means that the translations of these axioms are
derivable from the axioms of \ACPc\ with action renaming and guarded
recursion.
Before we make this fully precise, we have a closer look at the axioms
of \TApt+\REC+\TSU.

A proper axiom is an equation or a conditional equation.
In Tables~\ref{axioms-BTA}--\ref{axioms-spt}, we do not only find proper
axioms.
In addition to proper axioms, we find:
(i)~axiom schemas without side conditions;
(ii)~axiom schemas with syntactic side conditions;
(iii)~axiom schemas with semantic side conditions.
The axioms of \TApt+\REC+\TSU\ are obtained by replacing each axiom
schema by all its instances.
Owing to the presence of axiom schemas with semantic side conditions,
the axioms of \TApt+\REC+\TSU\ include proper axioms and axioms with
semantic side conditions.
Therefore, semantic side conditions take part in the translation of the
axioms as well.
The instances of TSU5, TSU6, and TSU7 are the only axioms of
\TApt+\REC+\TSU\ with semantic side conditions.
These semantic side conditions, being of the form $H(\seq{m}) = r$, are
looked upon as elements of $\ACond$.

Consider the set that consists of:
\begin{itemize}
\item
all equations $t_1 = t_2$, where $t_1$ and $t_2$ are terms of sort
$\Thr$;
\item
all conditional equations $E \Implies t_1 = t_2$, where $t_1$ and $t_2$
are terms of sort $\Thr$ and $E$ is a set of equations $t'_1 = t'_2$
where $t'_1$ and $t'_2$ are terms of sort $\Thr$;
\item
all expressions $t_1 = t_2 \smif \phi$, where $t_1$ and $t_2$ are terms
of sort $\Thr$ and $\phi \in \ACond$.
\end{itemize}
We define a translation function $\Mt{\ph}$ from this set to the set of
all equations of \ACPc\ with action renaming and guarded recursion as
follows:
\begin{ldispl}
\Mt{t_1 = t_2} \;\;=\;\; \Mt{t_1} = \Mt{t_2}\;,
\\
\Mt{E \Implies t_1 = t_2} \;\;=\;\;
\set{\Mt{t'_1} = \Mt{t'_2} \where t'_1 = t'_2 \,\in\, E} \Implies
\Mt{t_1} = \Mt{t_2}\;,
\\
\Mt{t_1 = t_2 \smif \phi} \;\;=\;\;
\ceval{h_{\Phi \union \set{\phi}}}(\Mt{t_1}) =
\ceval{h_{\Phi \union \set{\phi}}}(\Mt{t_2})\;,
\end{ldispl}%
where
\begin{ldispl}
\begin{aeqns}
\Phi & = &
\set{\AND{r \in \Replies}
      \lnot (H(\seq{m}) = r \land
            \OR{r' \in \Replies \diff \set{r}} H(\seq{m}) = r') \where
     H \!\in\! \Services, m \!\in\! \Meth}\;.
\end{aeqns}
\end{ldispl}%
Here $h_\Psi$ is a function on conditions of \ACPc\ that preserves
$\bot$, $\top$, $\bcompl$, $\join$ and $\meet$ and satisfies
$h_\Psi(\zeta) = \top$ iff $\zeta$ corresponds to a proposition
derivable from $\Psi$ and $h_\Psi(\zeta) = \bot$ iff $\bcompl \zeta$
corresponds to a proposition derivable from $\Psi$.%
\footnote
{Here we use ``corresponds to'' for the wordy ``is isomorphic to the
 equivalence class with respect to logical equivalence of''
 (see also~\cite{BM05a}).}
\begin{theorem}[Preservation]
Let $\nm{ax}$ be an axiom of \textup{\TApt+\REC+\TSU}.
Then $\Mt{\nm{ax}}$ is derivable from the axioms of \ACPc\ with action
renaming and guarded recursion.
\end{theorem}
\begin{proof}
The proof is straightforward.
In~\cite{BM05c}, we outline the proof for axiom TSU5.
The other axioms are proved in a similar way.
\qed
\end{proof}

\section{Execution Architectures for Fragmented Programs}
\label{sect-exe-arch}

An analytic execution architecture in the sense of~\cite{BP04a} is a
model of a hypothetical execution environment for sequential programs
that is designed for the purpose of explaining how a program may be
executed.
An analytic execution architecture makes explicit the interaction of a
program with the components of its execution environment.
The notion of analytic execution architecture defined in~\cite{BP04a} is
suited to sequential programs that have not been split into fragments.
In this section, we discuss analytic execution architectures suited to
sequential programs that have been split into fragments.

The notion of analytic execution architecture from~\cite{BP04a} is
defined in the setting of program algebra.
In~\cite{BL02a}, a thread extraction operation $\extr{\ph}$ is defined
which gives, for each program considered in program algebra, the thread
that is taken for the behaviour exhibited by the program on execution.
In the case of programs that have been split into fragments, additional
instructions for switching over execution to another program fragment
are needed.
We assume that a collection of program fragments between which execution
can be switched takes the form of a sequence, called an program fragment
vector, and that there is an additional instruction $\swo{i}$ for each
$i \in \Nat$.
Switching over execution to the $i$-th program fragment in the program
fragment vector is effected by executing the instruction $\swo{i}$.
If $i$ equals $0$ or $i$ is greater than the length of the program
fragment vector, execution of $\swo{i}$ results in deadlock.
We extend thread extraction as follows:
\begin{ldispl}
\extr{\swo{i}} = \Switch{i}\;, \qquad
\extr{\swo{i} \conc x} = \Switch{i}\;.
\end{ldispl}%

An analytic execution architecture for programs that have been split
into fragments consists of a component containing a program fragment, a
component containing a program fragment vector and a number of service
components.
The component containing a program fragment is capable of processing
instructions one at a time, issuing appropriate requests to service
components and awaiting replies from service components as described
in~\cite{BP04a} in so far as instructions other than switch-over
instructions are concerned.
This implies that, for each service component, there is a channel for
communication between the program fragment component and that service
component and that foci are used as names of those channels.
In the case of a switch-over instruction, the component containing a
program fragment is capable of loading the program fragment to which
execution must be switched from the component containing a program
fragment vector.

The analytic execution architecture made up of a component containing
the program fragment $P$, a component containing the program fragment
vector $\alpha = \seq{P_1} \concat \ldots \concat \seq{P_n}$, and
service components $H_1$, \ldots, $H_k$ with channels named $f_1$,
\ldots, $f_k$, respectively, is described by the thread
\begin{ldispl}
\spt{\extr{P}}{\seq{\extr{P_1}} \concat \ldots \concat \seq{\extr{P_n}}}
 \useop{f_1} H_1 \ldots  \useop{f_k} H_k\;.
\end{ldispl}%
In the case where instructions of the form $\swo{i}$ do not occur in
$P$,
\begin{ldispl}
\Mt{\spt{\extr{P}}
        {\seq{\extr{P_1}} \concat \ldots \concat \seq{\extr{P_n}}}
     \useop{f_1} H_1 \ldots  \useop{f_k} H_k}
\end{ldispl}%
agrees with the process-algebraic description given in~\cite{BP04a}
of the analytic execution
architecture made up of a component containing the program $P$ and
service components $H_1$, \ldots, $H_k$ with channels named $f_1$,
\ldots, $f_k$, respectively.

\section{Poly-Threaded Strategic Interleaving}
\label{sect-TAptsi}

In this section, we take up the extension of \TApt\ with a form of
interleaving suited for multi-threading.

Multi-threading refers to the concurrent existence of several threads
in a program under execution.
Multi-threading is provided by contemporary programming languages such
as Java~\cite{GJSB00a} and C\#~\cite{HWG03a}.
Arbitrary interleaving, on which \ACP~\cite{BW90}, CCS~\cite{Mil89} and
CSP~\cite{Hoa85} are based, is not an appropriate abstraction when
dealing with multi-threading.
In the case of multi-threading, some deterministic interleaving
strategy is used.
In~\cite{BM04c}, we introduced a number of plausible deterministic
interleaving strategies for multi-threading.
We proposed to use the phrase strategic interleaving for the more
constrained form of interleaving obtained by using such a strategy.
In this section, we consider the strategic interleaving of fragmented
program behaviours.

As in~\cite{BM04c}, it is assumed that the collection of threads to be
interleaved takes the form of a thread vector.
In this section, we only cover the simplest interleaving strategy for
fragmented program behaviours, namely pure \emph{cyclic interleaving}.
In the poly-threaded case, cyclic interleaving basically operates as
follows: at each stage of the interleaving, the first thread in the
thread vector gets a turn to perform a basic action or to switch over to
another thread and then the thread vector undergoes cyclic permutation.
We mean by cyclic permutation of a thread vector that the first thread
in the thread vector becomes the last one and all others move one
position to the left.
If one thread in the thread vector deadlocks, the whole does not
deadlock till all others have terminated or deadlocked.
An important property of cyclic interleaving is that it is fair, i.e.\
there will always come a next turn for all active threads.
Other plausible interleaving strategies are treated in~\cite{BM04c}.
They can also be adapted to the poly-threaded\linebreak[2] case.

The extension of \TApt\ with cyclic interleaving is called \TAptsi.
It has the sorts $\Thr$ and $\TV$ of \TApt.
To build terms of sort $\Thr$, \TAptsi\ has the constants and operators
of \TApt\ to build terms of sort $\Thr$ and in addition the following
operator:
\begin{iteml}
\item
the \emph{poly-threaded cyclic strategic interleaving} operator
$\funct{\pciop}{\TV \cprod \TV}{\Thr}$.
\end{iteml}
To build terms of sort $\TV$, \TAptsi\ has the constants and operators
of \TApt\ to build terms of sort $\TV$.

\TAptsi\ has the axioms of \TApt\ and in addition the axioms given in
Tables~\ref{axioms-pci} and~\ref{axioms-std}.%
%
\begin{table}[!t]
\caption{Axioms for poly-threaded cyclic interleaving}
\label{axioms-pci}
\begin{eqntbl}
\begin{saxcol}
\pci{\emptyseq}{\alpha} = \Stop                      & & \axiom{PCI1} \\
\pci{\seq{\Stop} \concat \beta}{\alpha} =
\pci{\beta}{\alpha}                                  & & \axiom{PCI2} \\
\pci{\seq{\DeadEnd} \concat \beta}{\alpha} =
\std{\pci{\beta}{\alpha}}                            & & \axiom{PCI3} \\
\pci{\seq{\pcc{x}{a}{y}} \concat \beta}{\alpha} =
\pcc{\pci{\beta \concat \seq{x}}{\alpha}}{a}
    {\pci{\beta \concat \seq{y}}{\alpha}}            & & \axiom{PCI4} \\
\pci{\seq{\Switch{i}} \concat \beta}
     {\seq{x_1} \concat \ldots \concat \seq{x_n}} =
{} \\ \quad
\tls.\init \bapf
\pci{\beta \concat \seq{x_i}}
    {\seq{x_1} \concat \ldots \concat \seq{x_n}}
                               & \mif 1 \leq i \leq n  & \axiom{PCI5} \\
\pci{\seq{\Switch{i}} \concat \beta}
    {\seq{x_1} \concat \ldots \concat \seq{x_n}} =
\std{\pci{\beta}{\seq{x_1} \concat \ldots \concat \seq{x_n}}}
                               & \mif i = 0 \lor i > n & \axiom{PCI6} \\
\pci{\seq{\Extern} \concat \beta}
    {\seq{x_1} \concat \ldots \concat \seq{x_k}} =
{} \\ \quad
 \choiceop{k}
  (\tls.\init \bapf
   \pci{\beta \concat \seq{x_1}}
       {\seq{x_1} \concat \ldots \concat \seq{x_k}}, \ldots,
{} \\ \quad \phantom{\choiceop{k}(}
   \tls.\init \bapf
   \pci{\beta \concat \seq{x_k}}
       {\seq{x_1} \concat \ldots \concat \seq{x_k}})
                                                     & & \axiom{PCI7} \\
\pci{\seq{\Extern} \concat \beta}{\emptyseq} =
\std{\pci{\beta}{\emptyseq}}                         & & \axiom{PCI8}
\end{saxcol}
\end{eqntbl}
\end{table}
\begin{table}[!t]
\caption{Axioms for deadlock at termination}
\label{axioms-std}
\begin{eqntbl}
\begin{axcol}
\std{\Stop} = \DeadEnd                                 & \axiom{S2D1} \\
\std{\DeadEnd} = \DeadEnd                              & \axiom{S2D2} \\
\std{\pcc{x}{a}{y}} = \pcc{\std{x}}{a}{\std{y}}        & \axiom{S2D3} \\
\std{\Switch{i}} = \Switch{i}                          & \axiom{S2D4} \\
\std{\Extern} = \Extern                                & \axiom{S2D5} \\
\std{\choiceop{k}(x_1,\ldots,x_{k})} =
\choiceop{k}(\std{x_1},\ldots,\std{x_{k}})             & \axiom{S2D6}
\end{axcol}
\end{eqntbl}
\end{table}
In these tables, $a$ stands for an arbitrary action from $\BActTau$.
The axioms from Table~\ref{axioms-pci} express that threads are
interleaved as described above.
In these axioms, the auxiliary \emph{deadlock at termination} operator
$\stdop$ occurs.
The axioms from Table~\ref{axioms-std} show that this operator serves to
turn termination into deadlock.

Guarded recursion and the use mechanism can be added to \TAptsi\ as they
are added to \BTA\ in Sections~\ref{sect-BTA} and~\ref{sect-TSU},
respectively.

\section{Poly-Threaded Distributed Strategic Interleaving}
\label{sect-TAptdsi}

In this section, we take up the extension of \TApt\ with a form of
interleaving suited for distributed multi-threading.

In order to deal with threads that are distributed over the nodes of a
network, it is assumed that there is a fixed but arbitrary finite set
$\Loc$ of \emph{locations} such that $\Loc \subseteq \Nat$.
The set $\LAct$ of \emph{located basic actions} is defined by
$\LAct = \set{l.a \where l \in \Loc \land a \in \BAct}$.
Henceforth, basic actions will also be called
\emph{unlocated basic actions}.
The members of $\LAct \union \set{l.\Tau \where l \in \Loc}$ are
referred to as \emph{located actions}.

Performing an unlocated action $a$ is taken as performing $a$ at a
location still to be fixed by the distributed interleaving strategy.
Performing a located action $l.a$ is taken as performing $a$ at location
$l$.

Threads that perform unlocated actions only are called \emph{unlocated}
threads and threads that perform located actions only are called
\emph{located} threads.
It is assumed that the collection of all threads that exist concurrently
at the same location takes the form of a sequence of unlocated threads,
called the \emph{local thread vector} at the location concerned.
It is also assumed that the collection of local thread vectors that
exist concurrently at the different locations takes the form of a
sequence of pairs, one for each location, consisting of a location and
the local thread vector at that location.
Such a sequence is called a \emph{distributed thread vector}.

In the distributed case, cyclic interleaving basically operates the same
as in the non-distributed case.
In the distributed case, we mean by cyclic permutation of a distributed
thread vector that the first thread in the first local thread vector
becomes the last thread in the first local thread vector, all other
threads in the first local thread vector move one position to the left,
the resulting local thread vector becomes the last local thread vector
in the distributed thread vector, and all other local thread vectors in
the distributed thread vector move one position to the left.

When discussing interleaving strategies on distributed thread vectors,
we use the term current thread to refer to the first thread in the first
local thread vector in a distributed thread vector and we use the term
current location to refer to the location at which the first local
thread vector in a distributed thread vector is.

The extension of \TApt\ with cyclic distributed interleaving is called
\TAptdsi. \linebreak[2]
It has the sorts $\Thr$ and $\TV$ of \TApt\ and in addition the
following sorts:
\begin{itemize}
\item
the sort $\LThr$ of \emph{located threads};
\item
the sort $\DTV$ of \emph{distributed thread vectors}.
\end{itemize}
To build terms of sort $\Thr$, \TAptdsi\ has the constants and operators
of \BTA\ and in addition the following operators:
\begin{itemize}
\item
for each $n \in \Nat$,
the \emph{migration postconditional composition} operator
$\funct{\pcc{\ph}{\MigThr{n}}{\ph}}{\Thr \cprod \Thr}{\Thr}$.
\end{itemize}
To build terms of sort $\TV$, \TAptdsi\ has the constants and operators
of \TApt\ to build terms of sort $\TV$.
To build terms of sort $\LThr$, \TAptdsi\ has the following constants
and operators:
\begin{itemize}
\item
the \emph{deadlock} constant $\const{\DeadEnd}{\LThr}$;
\item
the \emph{termination} constant $\const{\Stop}{\LThr}$;
\item
for each $l \in \Loc$ and $a \in \BActTau$,
the \emph{postconditional composition} operator
$\funct{\pcc{\ph}{l.a}{\ph}}{\LThr \cprod \LThr}{\LThr}$;
\item
the \emph{deadlock at termination} operator
$\funct{\stdop}{\LThr}{\LThr}$;
\item
the \emph{poly-threaded cyclic distributed strategic interleaving}
operator $\funct{\pciop}{\DTV \cprod \TV}{\LThr}$.
\end{itemize}
To build terms of sort $\DTV$, \TAptdsi\ has the following constants and
operators:
\begin{itemize}
\item
the \emph{empty distributed thread vector} constant
$\const{\emptyseq}{\DTV}$;
\item
for each $l \in \Loc$,
the \emph{singleton distributed thread vector} operator
$\funct{\atv{\ph}{l}}{\TV}{\DTV}$;
\item
\sloppy
the \emph{distributed thread vector concatenation} operator
$\funct{\concat}{\DTV \cprod \DTV}{\DTV}$.
\end{itemize}
Throughout the paper, we assume that there are infinitely many variables
of sort $\LThr$, including $u,v,w$, and infinitely many variables of
sort $\DTV$, including $\delta$.

We introduce \emph{located action prefixing} as an abbreviation:
$l.a \bapf p$, where $p$ is a term of sort $\LThr$, abbreviates
$\pcc{p}{l.a}{p}$.

The overloading of $\DeadEnd$, $\Stop$, $\emptyseq$ and $\concat$ could
be resolved, but we refrain from doing so because it is always clear
from the context which constant or operator is meant.

Essentially, the sort $\DTV$ includes all sequences of pairs consisting
of a location and a local thread vector.%
\footnote
{The singleton distributed thread vector operators involve an implicit
 pairing of their operand with a location.}
The ones that contain a unique pair for each location are the proper
distributed thread vectors in the sense that the cyclic distributed
interleaving strategy outlined above is intended for them.
Improper distributed thread vectors that do not contain duplicate pairs
for some location are needed in the axiomatization of this strategy.
Improper distributed thread vectors that do contain duplicate pairs for
some location appear to have more than one local thread vector
at the location concerned.
Their exclusion would make it necessary for concatenation of distributed
thread vectors to be turned into a partial operator.
The cyclic distributed interleaving strategy never turns a proper
distributed thread vector into an improper one or the other
way\linebreak[2] round.

The poly-threaded cyclic distributed strategic interleaving operator
serves for interleaving of the threads in a proper distributed thread
vector according to the strategy outlined above, but with support of
explicit thread migration.
In the case where a local thread vector of the form
$\seq{\pcc{p}{\MigThr{n}}{q}} \concat \gamma$ with $n \in \Loc$ is
encountered as the first local thread vector, $\gamma$ becomes the last
local thread vector in the distributed thread vector and $p$ is appended
to the local thread vector at location $n$.
If $n \not\in \Loc$, then $\gamma \concat \seq{q}$ becomes the last
local thread vector in the distributed thread vector.

In the axioms for cyclic distributed interleaving discussed below,
binary functions $\app{l}$ (where $l \in \Loc$) from unlocated threads
and distributed thread vectors to distributed thread vectors are used.
For each $l \in \Loc$, $\app{l}$ maps each unlocated thread $x$ and
distributed thread vector $\delta$ to the distributed thread vector
obtained by appending $x$ to the local thread vector at location $l$ in
$\delta$.
The functions $\app{l}$ are defined in Table~\ref{def-app}.%
\begin{table}[!t]
\caption{Definition of the functions $\app{l}$}
\label{def-app}
\begin{eqntbl}
\begin{seqncol}
\app{l}(x,\emptyseq) = \emptyseq \\
\app{l}(x,\atv{\gamma}{l'} \concat \delta) =
\atv{\gamma \concat \seq{x}}{l} \concat \delta & \mif l = l' \\
\app{l}(x,\atv{\gamma}{l'} \concat \delta) =
\atv{\gamma}{l'} \concat \app{l}(x,\delta)     & \mif l \neq l'
\end{seqncol}
\end{eqntbl}
\end{table}

\TAptdsi\ has the axioms of \TApt\ and in addition the axioms given in
Tables~\ref{axioms-pcc-lthr}, \ref{axioms-pcdi}
and~\ref{axioms-std-lthr}.%
\begin{table}[!t]
\caption{Axioms for postconditional composition}
\label{axioms-pcc-lthr}
\begin{eqntbl}
\begin{axcol}
\pcc{u}{l.\Tau}{v} = \pcc{u}{l.\Tau}{u}                & \axiom{LT1}
\end{axcol}
\end{eqntbl}
\end{table}%
%
\begin{table}[!t]
\caption{Axioms for poly-threaded cyclic distributed interleaving}
\label{axioms-pcdi}
\begin{eqntbl}
\begin{saxcol}
\pci{\emptyseq}{\alpha} = \Stop                     & & \axiom{PCDI1} \\
\pci{\atv{\emptyseq}{l_1} \concat \ldots \concat \atv{\emptyseq}{l_k}}
    {\alpha} =
\Stop                                               & & \axiom{PCDI2} \\
\pci{\atv{\emptyseq}{l} \concat \delta}{\alpha} =
\pci{\delta \concat \atv{\emptyseq}{l}}{\alpha}     & & \axiom{PCDI3} \\
\pci{\atv{\seq{\Stop} \concat \gamma}{l} \concat \delta}{\alpha} =
\pci{\delta \concat \atv{\gamma}{l}}{\alpha}        & & \axiom{PCDI4} \\
\pci{\atv{\seq{\DeadEnd} \concat \gamma}{l} \concat \delta}{\alpha} =
\std{\pci{\delta \concat \atv{\gamma}{l}}{\alpha}}  & & \axiom{PCDI5} \\
\pci{\atv{\seq{\pcc{x}{a}{y}} \concat \gamma}{l} \concat \delta}{\alpha}
 = {} \\ \quad
\pcc{\pci{\delta \concat \atv{\gamma \concat \seq{x}}{l}}{\alpha}}{l.a}
    {\pci{\delta \concat \atv{\gamma \concat \seq{y}}{l}}{\alpha}}
                                                    & & \axiom{PCDI6} \\
\pci{\atv{\seq{\Switch{i}} \concat \gamma}{l} \concat \delta}
    {\seq{x_1} \concat \ldots \concat \seq{x_n}} =
{} \\ \quad
l.\tls.\init \bapf \pci{\delta \concat \atv{\gamma \concat \seq{x_i}}{l}}
                       {\seq{x_1} \concat \ldots \concat \seq{x_n}}
                      & \hsp{-3}\mif 1 \leq i \leq n  & \axiom{PCDI7} \\
\pci{\atv{\seq{\Switch{i}} \concat \gamma}{l} \concat \delta}
    {\seq{x_1} \concat \ldots \concat \seq{x_n}} =
{} \\ \quad
\std{\pci{\delta \concat \atv{\gamma}{l}}
         {\seq{x_1} \concat \ldots \concat \seq{x_n}}}
                      & \hsp{-3}\mif i = 0 \lor i > n & \axiom{PCDI8} \\
\pci{\atv{\seq{\Extern} \concat \gamma}{l} \concat \delta}
    {\seq{x_1} \concat \ldots \concat \seq{x_k}} =
{} \\ \quad
\choiceop{k}
 (l.\tls.\init \bapf
  \pci{\delta \concat \atv{\gamma \concat \seq{x_1}}{l}}
      {\seq{x_1} \concat \ldots \concat \seq{x_k}}, \ldots,
{} \\ \quad \phantom{\choiceop{k}(}
  l.\tls.\init \bapf
  \pci{\delta \concat \atv{\gamma \concat \seq{x_k}}{l}}
      {\seq{x_1} \concat \ldots \concat \seq{x_k}}) & & \axiom{PCDI9} \\
\pci{\atv{\seq{\Extern} \concat \gamma}{l} \concat \delta}{\emptyseq} =
\std{\pci{\delta \concat \atv{\gamma}{l}}{\emptyseq}}
                                                    & & \axiom{PCDI10} \\
\pci{\atv{\seq{\pcc{x}{\MigThr{n}}{y}} \concat \gamma}{l} \concat \delta}
    {\alpha} =
l.\Tau \bapf \pci{\app{n}(x,\delta \concat \atv{\gamma}{l})}{\alpha}
                                & \mif n \in \Loc     & \axiom{PCDI11} \\
\pci{\atv{\seq{\pcc{x}{\MigThr{n}}{y}} \concat \gamma}{l} \concat \delta}
    {\alpha} =
l.\Tau \bapf \pci{\delta \concat \atv{\gamma \concat \seq{y}}{l}}{\alpha}
                                & \mif n \not\in \Loc & \axiom{PCDI12}
\end{saxcol}
\end{eqntbl}
\end{table}
\begin{table}[!t]
\caption{Axioms for deadlock at termination}
\label{axioms-std-lthr}
\begin{eqntbl}
\begin{axcol}
\std{\Stop} = \DeadEnd                                & \axiom{LS2D1} \\
\std{\DeadEnd} = \DeadEnd                             & \axiom{LS2D2} \\
\std{\pcc{u}{a}{v}} = \pcc{\std{u}}{a}{\std{v}}       & \axiom{LS2D3} \\
\std{\Switch{i}} = \Switch{i}                         & \axiom{LS2D4} \\
\std{\Extern} = \Extern                               & \axiom{LS2D5} \\
\std{\choiceop{k}(u_1,\ldots,u_{k})} =
\choiceop{k}(\std{u_1},\ldots,\std{u_{k}})            & \axiom{LS2D6}
\end{axcol}
\end{eqntbl}
\end{table}
In these tables, $a$ stands for an arbitrary action from $\BActTau$.
The axioms from Table~\ref{axioms-pcdi} express that threads are
interleaved as described above.
The axioms from Tables~\ref{axioms-pcc-lthr} and~\ref{axioms-std-lthr}
are the axioms from Tables~\ref{axioms-BTA} and~\ref{axioms-std} adapted
to located threads.

Guarded recursion and the use mechanism can be added to \TAptdsi\ as
they are added to \BTA\ in Sections~\ref{sect-BTA} and~\ref{sect-TSU},
respectively.

\section{Fragment Searching by Implicit Migration}
\label{sect-TAptdsifs}

In Section~\ref{sect-TAptdsi}, it was assumed that the same program
fragment behaviours are available at each location.
In the case where this assumption does not hold, distributed
interleaving strategies with implicit migration of threads to achieve
availability of fragments needed by the threads are plausible.
We say that such distributed interleaving strategies take care of
fragment searching.
In this section, we introduce a variation of the distributed
interleaving strategy from Section~\ref{sect-TAptdsi} with fragment
searching.
This results in a theory called \TAptdsifs.

It is assumed that there is a fixed but arbitrary set $\Ind$ of
\emph{fragment indices} such that $\Ind = [1,n]$ for some $n \in \Nat$.

In the case of the distributed interleaving strategy with fragment
searching, immediately after the current thread has performed an action,
implicit migration of that thread to another location may take place.
Whether migration really takes place, depends on the fragments present
at the current location.
The current thread is implicitly migrated if the following condition is
fulfilled: on its next turn, the current thread ought to switch over to
a fragment that is not present at the current location.
If this conditions is fulfilled, then the current thread will be
migrated to the first among the locations where the fragment concerned
is present.

To deal with that, we have to enrich distributed thread vectors.
The new distributed thread vectors are sequences of triples, one for
each location, consisting of a location, the local thread vector at that
location, and the set of all indices of fragments that are present at
that location.

\TAptdsifs\ has the same sorts as \TAptdsi.
To build terms of the sorts $\Thr$, $\TV$ and $\LThr$, \TAptdsifs\
has the same constants and operators as \TAptdsi.
To build terms of sort $\DTV$, \TAptdsifs\ has the following constants
and operators:
\begin{itemize}
\item
the \emph{empty distributed thread vector} constant
$\const{\emptyseq}{\DTV}$;
\item
for each $l \in \Loc$ and  $I \subseteq \Ind$,
the \emph{singleton distributed thread vector} operator
$\funct{\btv{\ph}{l}{I}}{\TV}{\DTV}$;
\item
the \emph{distributed thread vector concatenation} operator
$\funct{\concat}{\DTV \cprod \DTV}{\DTV}$.
\end{itemize}
That is, the operator $\atv{\ph}{l}$ is replaced by the operators
$\btv{\ph}{l}{I}$.

Essentially, the sort $\TV$ includes all sequences of unlocated threads.
These sequences may serve as local thread vectors and as fragment
vectors.
The sequences that contain a thread for each fragment index are proper
fragment vectors.
In the case of fragment vectors that contain more threads, there appear
to be inaccessible fragments and in the case of fragment vectors that
contain less threads, there appear to be disabled fragments.
Inaccessible fragments have no influence on the effectiveness of cyclic
distributed interleaving with fragment searching.
However, disabled fragments may lead to implicit migration to a location
where a switch-over on the next turn is not possible as well.
Should this case arise, the next turn will yield deadlock.

In the axioms for cyclic distributed interleaving with fragment
searching discussed below, binary functions $\appfs{l}$
(where $l \in \Loc$) from unlocated threads and distributed thread
vectors to distributed thread vectors are used which are similar to the
functions $\app{l}$ used in the axioms for cyclic distributed
interleaving without fragment searching given in
Section~\ref{sect-TAptdsi}.
The functions $\appfs{l}$ are defined in Table~\ref{def-app-fs}.%
\begin{table}[!t]
\caption{Definition of the functions $\appfs{l}$}
\label{def-app-fs}
\begin{eqntbl}
\begin{seqncol}
\appfs{l}(x,\emptyseq) = \emptyseq \\
\appfs{l}(x,\btv{\gamma}{l'}{I} \concat \delta) =
\btv{\gamma \concat \seq{x}}{l}{I} \concat \delta
& \mif l = l' \\
\appfs{l}(x,\btv{\gamma}{l'}{I} \concat \delta) =
\btv{\gamma}{l'}{I} \concat \appfs{l}(x,\delta)
& \mif l \neq l'
\end{seqncol}
\end{eqntbl}
\end{table}

Moreover, a unary function $\pv$ on distributed thread vectors is
used which permutes distributed thread vectors cyclicly with implicit
migration as outlined above.
The function $\pv$ is defined using two auxiliary functions:
\begin{itemize}
\item
a function $\iml'$ mapping each fragment index $i$, distributed thread
vector $\delta$ and location $l$ to the first location in $\delta$ at
which the fragment with index $i$ is present if the fragment concerned
is present anywhere, and location $l$ otherwise;
\item
a function $\iml$ mapping each non-empty distributed thread vector
$\delta$ to the first location in $\delta$ at which the fragment is
present to which the current thread ought to switch over on its next
turn if the current thread is in that circumstance and the fragment
concerned is present somewhere, and the current location otherwise.
\end{itemize}
The function $\pv$, as well as the auxiliary functions $\iml'$ and
$\iml$, are defined in Table~\ref{def-pvfs}.%
\begin{table}[!p]
\caption{Definition of the functions $\iml'$, $\iml$ and $\pv$}
\label{def-pvfs}
\begin{eqntbl}
\begin{eqncol}
\iml'(i,\emptyseq,l') = l' \\
\iml'(i,\btv{\gamma}{l}{I} \concat \delta,l') = l
 \hfill \mif i \in I \\
\iml'(i,\btv{\gamma}{l}{I} \concat \delta,l') = \iml'(i,\delta,l')
 \hfill \mif i \not\in I
\eqnsep
\iml(\btv{\emptyseq}{l}{I} \concat \delta) = l \\
\iml(\btv{\seq{\Stop} \concat \gamma}{l}{I} \concat \delta) =
l \\
\iml(\btv{\seq{\DeadEnd} \concat \gamma}{l}{I} \concat \delta) =
l \\
\iml
 (\btv{\seq{\pcc{x}{a}{y}} \concat \gamma}{l}{I} \concat \delta) =
l \\
\iml
 (\btv{\seq{\Switch{i}} \concat \gamma}{l}{I} \concat \delta) =
l
  \hfill \mif i \in I \\
\iml
 (\btv{\seq{\Switch{i}} \concat \gamma}{l}{I} \concat \delta) =
\iml'(i,\delta,l)
 \hfill \mif i \not\in I \\
\iml(\btv{\seq{\Extern} \concat \gamma}{l}{I} \concat \delta) =
l \\
\iml
 (\btv{\seq{\pcc{x}{\MigThr{n'}}{y}} \concat \gamma}{l}{I} \concat
  \delta) =
l
\eqnsep
\pv(\emptyseq) = \emptyseq \\
\pv(\btv{\emptyseq}{l}{I} \concat \delta) =
\btv{\emptyseq}{l}{I} \concat \delta \\
\iml(\btv{\seq{x} \concat \gamma}{l}{I} \concat \delta) = l'
 \Implies
\pv(\btv{\seq{x} \concat \gamma}{l}{I} \concat \delta) =
\appfs{l'}(x,\delta \concat \btv{\gamma}{l}{I})
\end{eqncol}
\end{eqntbl}
\end{table}

\TAptdsifs\ has the axioms of \TApt\ and in addition the axioms given in
Tables~\ref{axioms-pcc-lthr}, \ref{axioms-pcdi-im}
and~\ref{axioms-std-lthr}.%
%
\begin{table}[!p]
\caption{Axioms for cyclic distributed interleaving with fragment
 searching}
\label{axioms-pcdi-im}
\begin{eqntbl}
\begin{saxcol}
\pci{\emptyseq}{\alpha} = \Stop                   & & \axiom{PCDIfs1} \\
\pci{\btv{\emptyseq}{l_1}{I} \concat \ldots \concat
     \btv{\emptyseq}{l_k}{I}}
    {\alpha} =
\Stop                                             & & \axiom{PCDIfs2} \\
\pci{\btv{\emptyseq}{l}{I} \concat \delta}{\alpha} =
\pci{\delta \concat \btv{\emptyseq}{l}{I}}{\alpha} & & \axiom{PCDIfs3} \\
\pci{\btv{\seq{\Stop} \concat \gamma}{l}{I} \concat \delta}{\alpha} =
\pci{\delta \concat \btv{\gamma}{l}{I}}{\alpha}   & & \axiom{PCDIfs4} \\
\pci{\btv{\seq{\DeadEnd} \concat \gamma}{l}{I} \concat \delta}{\alpha} =
\std{\pci{\delta \concat \btv{\gamma}{l}{I}}{\alpha}} & & \axiom{PCDIfs5} \\
\pci{\btv{\seq{\pcc{x}{a}{y}} \concat \gamma}{l}{I} \concat \delta}{\alpha}
 = {} \\ \quad
\pcc{\pci{\pv(\btv{\seq{x} \concat \gamma}{l}{I} \concat \delta)}{\alpha}}
    {l.a}
    {\pci{\pv(\btv{\seq{y} \concat \gamma}{l}{I} \concat \delta)}{\alpha}}
                                                  & & \axiom{PCDIfs6} \\
\pci{\btv{\seq{\Switch{i}} \concat \gamma}{l}{I} \concat \delta}
    {\seq{x_1} \concat \ldots \concat \seq{x_n}} =
{} \\ \quad
l.\tls.\init \bapf
\pci{\pv(\btv{\seq{x_i} \concat \gamma}{l}{I} \concat \delta)}
    {\seq{x_1} \concat \ldots \concat \seq{x_n}}
        & \hsp{-2.75} \mif i \in I \inter [1,n]     & \axiom{PCDIfs7} \\
\pci{\btv{\seq{\Switch{i}} \concat \gamma}{l}{I} \concat \delta}
    {\seq{x_1} \concat \ldots \concat \seq{x_n}} =
{} \\ \quad
\std{\pci{\delta \concat \btv{\gamma}{l}{I}}
         {\seq{x_1} \concat \ldots \concat \seq{x_n}}}
        & \hsp{-2.75} \mif i \not\in I \inter [1,n] & \axiom{PCDIfs8} \\
\pci{\btv{\seq{\Extern} \concat \gamma}{l}{I} \concat \delta}
    {\seq{x_1} \concat \ldots \concat \seq{x_k}} =
{} \\ \quad
\choiceop{k}
 (l.\tls.\init \bapf
  \pci{\pv(\btv{\seq{x_1} \concat \gamma}{l}{I} \concat \delta)}
      {\seq{x_1} \concat \ldots \concat \seq{x_k}}, \ldots,
{} \\ \quad \phantom{\choiceop{k}(}
  l.\tls.\init \bapf
  \pci{\pv(\btv{\seq{x_k} \concat \gamma}{l}{I} \concat \delta)}
      {\seq{x_1} \concat \ldots \concat \seq{x_k}}) & & \axiom{PCDIfs9} \\
\pci{\btv{\seq{\Extern} \concat \gamma}{l}{I} \concat \delta}{\emptyseq} =
\std{\pci{\delta \concat \btv{\gamma}{l}{I}}{\emptyseq}}
                                                    & & \axiom{PCDIfs10} \\
\pci{\btv{\seq{\pcc{x}{\MigThr{n}}{y}} \concat \gamma}{l}{I} \concat \delta}
    {\alpha} =
l.\Tau \bapf
\pci{\appfs{n}(x,\delta \concat \btv{\gamma}{l}{I})}
    {\alpha}                    & \mif n \in \Loc     & \axiom{PCDIfs11} \\
\pci{\btv{\seq{\pcc{x}{\MigThr{n}}{y}} \concat \gamma}{l}{I} \concat \delta}
    {\alpha} =
l.\Tau \bapf
\pci{\pv(\btv{\seq{y} \concat \gamma}{l}{I} \concat \delta)}
    {\alpha}                    & \mif n \not\in \Loc & \axiom{PCDIfs12}
\end{saxcol}
\end{eqntbl}
\end{table}

Guarded recursion and the use mechanism can be added to \TAptdsifs\ as
they are added to \BTA\ in Sections~\ref{sect-BTA} and~\ref{sect-TSU},
respectively.

\section{Conclusions}
\label{sect-concl}

We have developed a theory of the behaviours exhibited by sequential
programs on execution that covers the case where the programs have been
split into fragments and have used it to describe analytic execution
architectures suited for such programs.
It happens that the resulting description is terse.
We have also shown that threads and services as considered in this
theory can be viewed as processes that are definable over an extension
of \ACP\ with conditions.
Threads and services are introduced for pragmatic reasons only:
describing them as general processes is awkward.
For example, the description of analytic execution architectures suited
for programs that have been split into fragments would no longer be
terse if \ACP\ with conditions had been used.

We have also taken up the extension of the theory developed to the case
where the steps of fragmented program behaviours are interleaved in the
ways of non-distributed and distributed multi-threading.
This work can be further elaborated on the lines of~\cite{BM07a} to
cover issues such as prevention from migration for threads that
keep locks on shared services, load balancing by means of implicit
migration, and the use of implicit migration to achieve availability of
services needed by threads.

The object pursued with the line of research that we have carried on
with this paper is the development of a theoretical understanding of
the concepts sequential program and sequential program behaviour.
We regard the work presented in this paper also as a preparatory step in
the development of a theoretical understanding of the concept operating
system.

\bibliographystyle{spmpsci}
\bibliography{TA}

\end{document}